\documentclass[12pt]{article}
\usepackage{amsmath,amssymb,amsthm}
\RequirePackage[dvips]{graphicx}
\usepackage{xcolor,float}
\usepackage{dsfont}
\usepackage{tikz,multirow,booktabs,adjustbox,hyperref}
\usepackage[superscript]{cite}
\usepackage{cite}
\usepackage{xcolor}
\definecolor{AmberDawn}{RGB}{255, 191, 0}
\usepackage[table,xcdraw,dvipsnames]{xcolor}
\definecolor{LightSlateBlue}{rgb}{0.52, 0.44, 1.0}
\definecolor{DarkSeaGreen}{rgb}{0.56, 0.74, 0.56}
\definecolor{LightSteelBlue}{rgb}{0.69, 0.77, 0.87}
\definecolor{MediumPurple}{rgb}{0.58, 0.44, 0.86}
\definecolor{MediumAquamarine}{rgb}{0.4, 0.8, 0.67}
\definecolor{LightSkyBlue}{rgb}{0.53, 0.81, 0.98}
\definecolor{SteelBlue}{rgb}{0.27, 0.51, 0.71}
\definecolor{Wheat}{rgb}{0.96, 0.87, 0.7}
\definecolor{LightCyan}{rgb}{0.88, 1.0, 1.0}
\usepackage[table,xcdraw,dvipsnames]{xcolor}

\definecolor{Khaki}{rgb}{0.94, 0.9, 0.55}
\definecolor{DarkOrchid}{rgb}{0.6, 0.2, 0.8}
\definecolor{MediumSpringGreen}{rgb}{0.0, 0.98, 0.6}
\definecolor{MidnightBlue}{rgb}{0.1, 0.1, 0.44}
\definecolor{MintCream}{rgb}{0.96, 1.0, 0.98}
\definecolor{LightGreen}{rgb}{0.56, 0.93, 0.56}
\definecolor{PowderBlue}{rgb}{0.69, 0.88, 0.9}
\definecolor{LightPink}{rgb}{1.0, 0.71, 0.76}
\definecolor{AntiqueWhite}{rgb}{0.98, 0.92, 0.84}
\definecolor{Beige}{rgb}{0.96, 0.96, 0.86}
\definecolor{PaleGreen}{rgb}{0.6, 0.98, 0.6}
\definecolor{Ivory}{rgb}{1.0, 1.0, 0.94}
\definecolor{Azure}{rgb}{0.94, 1.0, 1.0}
\definecolor{FloralWhite}{rgb}{1.0, 0.98, 0.94}
\definecolor{AliceBlue}{rgb}{0.94, 0.97, 1.0}
\definecolor{MistyRose}{rgb}{1.0, 0.89, 0.88}
\definecolor{GhostWhite}{rgb}{0.97, 0.97, 1.0}
\definecolor{LightYellow}{rgb}{1.0, 1.0, 0.88}
\definecolor{Linen}{rgb}{0.98, 0.94, 0.9}
\definecolor{OldLace}{rgb}{0.99, 0.96, 0.9}
\definecolor{Honeydew}{rgb}{0.94, 1.0, 0.94}
\definecolor{Seashell}{rgb}{1.0, 0.96, 0.93}
\definecolor{Snow}{rgb}{1.0, 0.98, 0.98}
\definecolor{Gainsboro}{rgb}{0.86, 0.86, 0.86}
\usepackage[xcdraw,table,dvipsnames]{xcolor}
\definecolor{SlateGray}{rgb}{0.44, 0.5, 0.56}
\definecolor{LightCoral}{rgb}{0.94, 0.5, 0.5}
\usepackage{makecell}
\theoremstyle{plain}
\newtheorem{theorem}{Theorem}[section]

\newtheorem{lemma}[theorem]{Lemma}

\theoremstyle{definition}

\newtheorem{claim}{Claim}

\begin{document}

\begin{center}
{\Large \bf  Resolving Open Problems on the Hyper-Zagreb\\[2mm] Index and its Chemical Applications }

 \vspace{8mm}

 {\bf Kinkar Chandra Das, Jayanta Bera\footnote{Corresponding author.}}

 \vspace{9mm}

 \baselineskip=0.20in

 {\it Department of Mathematics, Sungkyunkwan University, \\
 Suwon 16419, Republic of Korea\/} \\[2mm]
 E-mail: {\tt  kinkardas2003@gmail.com,~jayantabera@g.skku.edu}

\vspace{4mm}

 \end{center}

 \vspace{5mm}

 \baselineskip=0.20in

\begin{abstract}

Topological indices are numerical invariants derived from molecular graphs and play an important role in characterizing chemical compounds and predicting their properties. Among the earliest descriptors are the classical Zagreb indices introduced by Gutman and Trinajsti\'c in 1972. A more recent development is the hyper-Zagreb index ($HM$), defined as $HM(G)=\sum_{v_i v_j\in E(G)}(d_i+d_j)^2$, where $d_i$ denotes the degree of vertex $v_i$. In 2023, Hayat et al. posed an open problem concerning bounds on the $HM$ index under fixed vertex-connectivity or edge-connectivity, along with the characterization of the corresponding extremal graphs. In this work, the problem is resolved by determining the extremal graphs that maximize $HM$ index under these constraints. The investigation is further extended to several additional extremal problems, including graphs with a given number of leaves, chromatic number, and independence number. The associated extremal graphs are identified in each case. In addition, the chemical relevance of $HM$ is examined through QSPR studies. Finally, the conclusion is presented.

\bigskip

\end{abstract}
\baselineskip=0.27in

\section{Introduction}

Mathematical chemistry is a multidisciplinary area that employs mathematical tools to investigate chemical structures, reactions, and properties. In particular, researchers apply graph-theoretic methods to analyze molecular systems with enhanced accuracy and efficiency by modeling molecules as graphs where atoms serve as vertices and chemical bonds as edges.

Topological indices are a key concept in this approach. These numerical invariants are obtained from the molecular graph and capture essential structural characteristics such as symmetry and connectivity. These invariants are widely utilized in quantitative structure–property relationship (QSPR) studies to  explain chemical properties by providing a compact mathematical representation of molecular structure. Through QSPR models, topological indices help to link molecular structure to experimental properties, aiding in drug design, materials discovery, and toxicity prediction.

Since Wiener's groundbreaking work in 1947,\cite{wiener47} many topological indices have been used in academic research to describe molecular structures using parameters like degree, eccentricity, and distance.\cite{Heliyon2020,dong11,das1666} Degree-based descriptors have been especially important among these, having a significant impact on both theoretical research and real-world applications.\cite{Shanmukha244,lucic09} In 1972, Gutman and Trinajsti\'c described the Zagreb indices,\cite{gut72} which are among the earliest degree-based topological indices. For a graph $G$ with vertex set $V(G) = \lbrace v_1, v_2, \ldots,
v_n \rbrace$ and edge set $E(G)$, the first Zagreb index ($M_1$) and the second Zagreb index ($M_2$ ) are formulated as follows:
\begin{eqnarray*}\label{H}
M_1(G)&=&\displaystyle{\sum\limits_{v_{i}\,v_{j} \in
E(G)}\,(d_i+d_j)},\quad
M_2(G)=\displaystyle{\sum\limits_{v_{i}\,v_{j} \in
E(G)}{d_i\,d_j}},
 \end{eqnarray*}
 where $d_i$ represents the degree of the vertex $v_i$. These indices have been very helpful in analyzing different molecular properties, such as complexity, chirality, and hetero-systems. These indices have been the subject of a substantial body of literature that covers both mathematical and chemical aspects. For example, Liu and Gutman gave some bounds on the $M_1$ and $M_2$ indices for connected graphs.\cite{liu066} Das and Gutman explored the $M_2$ index in the context of a graph and its complement.\cite{das044} Das compared the $M_1$ and $M_2$ indices.\cite{das100} Gutman et al. provided a comprehensive solution to the problem of finding the lowest value of $M_1$ among all graphs with a fixed number of pendant vertices.\cite{gut15} Liu et al. described the maximal graphs for $M_1$ within a class of graphs characterized by a certain vertex-connectivity.\cite{liu177} Over the past 30 years, Gutman and Das presented a comprehensive survey of the literature on $M_1$.\cite{gut04} For further results on the $M_1$ and $M_2$ indices, the reader may consult the relevant references.\cite{liu102,gutman13,nikolic03,xu11}.

 Recently, Shirdel et al. presented the hyper-Zagreb index ($HM$),\cite{shirdel13} which is a refinement of the first Zagreb index. The $HM$ index 
 is defined as
 \begin{eqnarray*}
 HM(G)=\displaystyle{\sum\limits_{v_{i}\,v_{j} \in
E(G)}{(d_i+d_j)^2}}.
 \end{eqnarray*}
Mirajkar et al. found that $HM$ provides higher accuracy in modelling $DHVAP$, $HVAP$, entropy, and acentric factor of octanes.\cite{mirajkar21}
The effectiveness of the $HM$ index in explaining some physicochemical properties of monocarboxylic acids was examined in the literature.\cite{hayat2023}
Elumalai et al. provided new bounds for the $HM$ index of connected graphs.\cite{elumalai18,elumalai19,elumalai22}
Liu and Tang established the maximum value of the $HM$ index for trees, unicyclic graphs, and bicyclic graphs with a specified matching number, and further characterized the corresponding extremal graphs.\cite{liu200}
Raza et al. examined the extremal trees with respect to the $HM$ index under constraints on maximum degree.\cite{raza22}
For further results on the $HM$ index, the reader may consult the relevant references.\cite{vetrik,zhong20}
Hayat et al. posed some interesting open problems on $HM$,\cite{hayat2023} as reported below:

\begin{itemize}
    \item \textbf{Problem:} Determine sharp bounds for $HM$ of graphs with specified   vertex-connectivity or edge-connectivity, and classify  extremal graphs.
   \item \textbf{Problem:} Find sharp bounds with respect to $HM$ of graphs with prescribed diameter or matching number, and identify the extremal graphs.
\end{itemize}

    In recent studies, maximal graphs with respect to $HM$ for a fixed graph order and either a prescribed diameter or matching number were investigated.\cite{vetrik23,walia24} In this paper, we classify graphs with the highest $HM$ index by considering those with prescribed vertex-connectivity, edge-connectivity, number of leaves, chromatic number, or independence number. Furthermore, we analyze the contribution of the $HM$  index to the modeling of structure–property relationships of nonane isomers.
\vspace*{3mm}

\section{Mathematical results on $HM$ of graphs.}
 In this section, we explore the mathematical properties of $HM$ for graphs. 
 \begin{lemma} \label{a1} $(i)$ $HM(G)>HM(G-e)$ for any edge $e\in E(G)$, and $(ii)$ $HM(G+e)>HM(G)$ for $e\notin E(G)$.
 \end{lemma}

The split graph $SP(n,\beta)$ is a graph of order $n$ with independence number $\beta$, defined as
    $$SP(n,\beta)=\overline{K_{\beta}\vee (n-\beta)\,K_1},$$
 where $\vee$ denotes the join of graphs and the overline indicates the graph
 complement. Applying Lemma \ref{a1}, it is simple to observe the following result:
  \begin{theorem} Let $G$ be a graph of order $n$ with independence number $\beta$. Then
    $$HM(G)\leq 2\,(n-\beta)\,(n-\beta-1)\,(n-1)^2+\beta\,(n-\beta)\,(2n-\beta-1)^2$$
  with equality iff $G\cong SP(n,\beta)$.
  \end{theorem}
  
The \textbf{Tur\'an graph}, represented by $T(n,r)$, is a complete $r$-partite graph on $n$ vertices. More specifically, the vertex set is divided into $r$ disjoint independent parts such that the sizes of any two parts differ by at most one. That is, if the part sizes are $n_1, n_2, \dots, n_r$, then
\[
\sum_{i=1}^r n_i = n \quad \text{and} \quad |n_i - n_j| \leq 1 \quad \text{for all } i, j.
\]
Each vertex is directly connected to all vertices outside its own part. We write $K_{n_1, n_2, \dots, n_r}$ to represent a complete $r$-partite graph with parts of sizes $n_1, n_2, \dots, n_r$.
We now describe the maximal graphs for $HM$ within a class of graphs identified by a specific chromatic number.
 \begin{theorem} Let $G$ be a graph of order $n$ with chromatic number $r$. Then
 \begin{align}
 HM(G) &\leq t\,(r-t)\,\left\lfloor \frac{n}{r} \right\rfloor
 \left\lceil \frac{n}{r} \right\rceil \left(2n - \left\lceil
 \frac{n}{r} \right\rceil- \left\lfloor \frac{n}{r}
 \right\rfloor\right)^2+4\,\binom{r-t}{2} \left\lfloor \frac{n}{r}
 \right\rfloor^2 \left(n -\left\lfloor \frac{n}{r} \right\rfloor\right)^2\nonumber\\
 &\quad\quad + 4\,\binom{t}{2} \left\lceil \frac{n}{r} \right\rceil^2
 \left(n - \left\lceil \frac{n}{r} \right\rceil\right)^2,\label{1t1}
 \end{align}
 where $n=r\left\lfloor \frac{n}{r}\right\rfloor+t$, $0\leq t<r$.
 Furthermore, the equality occurs in (\ref{1t1}) iff $G \cong T(n,r)$.
 \end{theorem}

 \begin{proof} Since $G$ has a graph of order $n$ with chromatic number $r$, $G$ can be partitioned into $r$ independent sets of sizes
 $n_1, n_2, \ldots, n_r$. Without loss of generality, suppose $n_1 \geq n_2 \geq \ldots \geq n_r\geq 1$. According to Lemma \ref{a1},
 we say that $HM(G)\leq HM(K_{n_1,n_2,\ldots,n_r})$ with equality iff $G\cong K_{n_1, n_2, \ldots,
 n_r}$. If $n_1-n_r\leq 1$, then $G\cong T(n,r)$ and so the equality occurs in (\ref{1t1}). Otherwise, $n_1-n_r\geq
 2$. For $r=2$, we obtain
 \begin{align*}
 HM(K_{n_1-1,n_2+1})-HM(K_{n_1,n_2})&=(n_1-1)\,(n_2+1)\,(2n-n_1-n_2)^2-n_1\,n_2\,(2n-n_1-n_2)^2\\
                                  &=(n_1-n_r-1)\,n^2>0
 \end{align*}
 as $n_1+n_2=n$. For $r\geq 3$, one can easily see that
 \begin{align*}
 &HM(K_{n_1-1, n_2, \ldots, n_r+1})-HM(K_{n_1, n_2, \ldots, n_r})\\
 =&(n_1-1)\,(n_r+1)\,(2n-n_1-n_r)^2-n_1\,n_r\,(2n-n_1-n_r)^2+(n_1-1)\,\sum\limits^{r-1}_{i=2}\,n_i\,(2n-n_1-n_i+1)^2\\
 &-n_1\,\sum\limits^{r-1}_{i=2}\,n_i\,(2n-n_1-n_i)^2+(n_r+1)\,\sum\limits^{r-1}_{i=2}\,n_i\,(2n-n_r-n_i-1)^2-n_r\,\sum\limits^{r-1}_{i=2}\,n_i\,
 (2n-n_r-n_i)^2\\[2mm]
 =&(n_1-n_r-1)\,(2n-n_1-n_r)^2-\sum\limits^{r-1}_{i=2}\,n_i\,(2n-n_1-n_i+1)^2+\sum\limits^{r-1}_{i=2}\,n_i\,(2n-n_r-n_i-1)^2\\[2mm]
 &+n_1\,\sum\limits^{r-1}_{i=2}\,n_i\,\Big((2n-n_1-n_i+1)^2-(2n-n_1-n_i)^2\Big)\\[2mm]
 &-n_r\,\sum\limits^{r-1}_{i=2}\,n_i\,\Big((2n-n_r-n_i)^2-(2n-n_r-n_i-1)^2\Big)\\[2mm]
 =&(n_1-n_r-1)\,(2n-n_1-n_r)^2+\sum\limits^{r-1}_{i=2}\,n_i\,\Big((2n-n_r-n_i-1)^2-(2n-n_1-n_i+1)^2\Big)\\[2mm]
 &+n_1\,\sum\limits^{r-1}_{i=2}\,n_i\,\Big(2\,(2n-n_1-n_i)+1\Big)-n_r\,\sum\limits^{r-1}_{i=2}\,n_i\,\Big(2\,(2n-n_r-n_i)-1\Big)\\[2mm]
 >&\sum\limits^{r-1}_{i=2}\,n_i\,\Big(2n_1\,(2n-n_1-n_i)-2n_r\,(2n-n_r-n_i\Big)+\sum\limits^{r-1}_{i=2}\,n_i\,(n_1-n_r)\\[2mm]
 >&2\,\sum\limits^{r-1}_{i=2}\,n_i\,(n_1-n_r)\,(2n-n_i-n_1-n_r)\geq 0.
 \end{align*}
 Thus we have $HM(K_{n_1-1, n_2, \ldots, n_r+1})>HM(K_{n_1, n_2,\ldots,n_r})$. From this result, we obtain the following:
 \begin{align*}
 HM(K_{n_1, n_2, \ldots, n_r})<HM(K_{n_1-1, n_2, \ldots,
 n_r+1})<\cdots<HM(T(n,r)).
 \end{align*}
 The inequality in (\ref{1t1}) holds strictly. Thus, the theorem is proved.
 \end{proof}

  Let $S(n_1,\,n_2, \ldots,\,n_{n-p})$ be a graph with 
$n$ vertices, consisting of a clique $K_{n-p}$ with $n-p$ vertices and $p$ leaves, where
 each $n_i$ represents the number of leaves attached to the $i$-th vertex of the
 clique $K_{n-p}$ and $ \sum_{i=1}^{n-p}\, n_i=p$ with $n_i\geq 0$ for $1\leq i\leq n-p$. A special case of this graph is when all $n_i = 0$,
 i.e., no leaves are attached. In this case, the graph is simply the complete graph $K_n$, and we write: $S(0,\,0, \ldots,\,0)\cong K_n.$
 The pineapple graph $K_{n,p}$ is a graph obtained by identifying one vertex of $K_{n-p}$ with central vertex of
 star $K_{1,p}$, that is, $K_{n,p}\cong S(p,\,\underbrace{0, \ldots,\,0}_{n-p-1})$. One can easily see that
    $$HM(K_{n,p})=(n-p-1)\,(2n-p-2)^2+4\,{n-p-1\choose 2}\,(n-p-1)^2+p\,n^2.$$
 We now present an upper bound on $HM(G)$ of graph $G$ in terms of
 $n$ and the number of leaves $p$, and identify the maximal graphs.
 \begin{theorem}\label{p1} Let $G$ be a graph of order $n$ with $p\,(\leq n-2)$ leaves. Then
 \begin{align}
 HM(G)\leq (n-p-1)\,(2n-p-2)^2+4\,{n-p-1\choose 2}\,(n-p-1)^2+p\,n^2\label{pkd1}
 \end{align}
 with equality iff $G\cong K_{n,p}$.
 \end{theorem}

 \begin{proof} If $p=0$, then from Lemma \ref{a1}, it is obvious that $HM(G)\leq HM(K_n)=2n\,(n-1)^3$ with equality iff
 $G\cong K_n$. Otherwise, $p\geq 1$. Let $H$ be a graph with 
$n$ vertices and the number of leaves $p\,(\geq 1)$ such that the hyper-Zagreb
 index $HM(H)$ is maximum. So $HM(G)\leq HM(H)$ with equality iff $G\cong H$. Let $q=n-p$. Then we have $q\geq 2$.
 We take a set of leaves $S\subseteq V(H)$ with
 $|S|=p$. By Lemma \ref{a1}, we obtain $H[V(H)-S]\cong K_{n-p}$. Thus the graph $H$ has the form $H\cong S(n_1,\,n_2, \ldots,\,n_{n-p})$. If
 $n_1=p$, then $n_2=n_3=\cdots=n_{n-p}=0$ and hence
 \begin{align*}
 HM(G)&\leq HM(H)=HM(S(n_1,\,n_2, \ldots,\,n_{n-p}))=HM(K_{n,p})\\
 &=(n-p-1)\,(2n-p-2)^2+4\,{n-p-1\choose 2}\,(n-p-1)^2+p\,n^2
 \end{align*}
 with equality iff $G\cong K_{n,p}$. Otherwise, $n_1<p$.
 Thus we have $n_2\geq 1$. Without loss of generality, we suppose
 that there is a positive integer $t\,(\geq 2)$ with $n_t\geq 1$ and $n_{t+1}=n_{t+2}=\cdots=n_{n-p}=0$.

 \vspace*{3mm}

 \noindent
 \begin{claim} \label{h11}
  $$HM\Big(S(n_1+1,\,n_2,\ldots,\,n_t-1,\ldots,\,n_{n-p})\Big)>HM\Big(S(n_1,\,n_2,\ldots,\,n_t,\ldots,\,n_{n-p})\Big).$$
 \end{claim}

 \noindent
 {\bf Proof of Claim \ref{h11}.} Now,
 \begin{align*}
 &HM\Big(S(n_1+1,\,n_2,\ldots,\,n_t-1,\ldots,\,n_{n-p})\Big)-HM\Big(S(n_1,\,n_2,\ldots,\,n_t,\ldots,\,n_{n-p})\Big)\\
 =&(n_1+1)\,(n+n_1-p+1)^2-n_1\,(n+n_1-p)^2+(n_t-1)\,(n+n_t-p-1)^2-n_t\,(n+n_t-p)^2\\
 &+\sum_{\substack{i = 2 \\ i \ne t}}^{n - p}
\,\Big((2n-2p+n_1+n_i-1)^2-(2n-2p+n_1+n_i-2)^2\Big)\\
 &+\sum_{\substack{i = 2 \\ i \ne t}}^{n - p}\,\Big((2n-2p+n_t+n_i-3)^2-(2n-2p+n_t+n_i-2)^2\Big)\\
 =&\Big((n+n_1-p+1)^2-(n+n_1-p)^2\Big)\,n_1+\Big((n+n_t-p-1)^2-(n+n_t-p)^2\Big)\,(n_t-1)\\
 &+(n+n_1-p+1)^2-(n+n_t-p)^2+\sum_{\substack{i = 2 \\ i \ne t}}^{n - p}\,\Big(2\,(2n-2p+n_1+n_i-2)+1\Big)\\
 &-\sum_{\substack{i = 2 \\ i \ne t}}^{n - p}\,\Big(2\,(2n-2p+n_t+n_i-2)-1\Big)\\
 &=(2n+2n_1-2p+1)\,n_1-(2n+2n_t-2p-1)\,(n_t-1)+(2n+2n_1-2p+1)\\
 &~~~~~~~~~~~~~~~~~+2\,\sum_{\substack{i = 2 \\ i \ne t}}^{n - p}\,\Big(n_1-n_t+1\Big)\\
 &>0
 \end{align*}
 as $n-p\geq 2$. Thus, the {\bf Claim \ref{h11}} is proved.

 \vspace*{3mm}

 By applying {\bf Claim \ref{h11}} repeatedly, we get the following inequalities:
 \begin{align*}
 &HM\Big(S(n_1,\,n_2,\ldots,\,n_t,\ldots,\,n_{n-p})\Big)<HM\Big(S(n_1+1,\,n_2,\ldots,\,n_t-1,\ldots,\,n_{n-p})\Big)<\cdots\\
    &<HM\Big(S(p-1,\,1,\underbrace{0,\ldots,\,0}_{n-p-2})\Big)<HM\Big(S(p,\underbrace{0,\ldots,\,0}_{n-p-1})\Big)=HM\Big(K_{n,p}\Big).
 \end{align*}
 Using the above result, we obtain
   $$HM(G)\leq HM(H)=HM\Big(S(n_1,n_2,\ldots,n_{n-p})\Big)<HM\Big(K_{n,p}\Big).$$
 The result in (\ref{pkd1}) holds strictly. Thus, the theorem is proved.
 \end{proof}

We now prove that for a graph $G$ with $n$ vertices and vertex connectivity $k$, the graph $(K_1 \cup K_{n-k-1}) \vee K_k$ maximizes the value $HM(G)$.
 \begin{theorem}\label{pp1} Let $G$ be a graph of order $n$ with vertex connectivity $k$. Then
 \begin{align}
 HM(G)\leq 4\,{k\choose 2}\,(n-1)^2+4\,{n-k-1\choose 2}\,(n-2)^2+k\,(n-k-1)\,(2n-3)^2+k\,(n+k-1)^2\label{p1kd1}
 \end{align}
 with equality iff $G\cong (K_1\cup K_{n-k-1})\vee K_k$.
 \end{theorem}

 \begin{proof} Let $H$ be a graph of order $n$ and vertex connectivity $k$ such that the hyper-Zagreb index $\mathrm{HM}(H)$ is
 maximum. So $HM(G)\leq HM(H)$ with equality iff $G\cong H$.
 Consider a vertex cut $S\subseteq V(H)$ with $|S|=k$, for which the removal of $S$ disconnects the graph $H$. Since $H$ maximizes
 the hyper-Zagreb index $\mathrm{HM}(H)$ among all such graphs of order $n$ with vertex connectivity $k$, it follows that $H-S$ must have
 exactly two connected components. Denote these components by $H_1$ and $H_2$, with $|V(H_1)| = k_1$ and $|V(H_2)| =n-k-k_1$. Without loss
 of generality, we suppose that $k_1 \leq n-k-k_1$, that is, $1\leq k_1\leq \lfloor\frac{n-k}{2}\rfloor$ ($H_1$ is the smaller of the
 two components). By Lemma \ref{a1}, we can state that $H_1 \cong K_{k_1}, \quad H[S] \cong K_{k}, \quad H_2 \cong
 K_{n-k-k_1}$ and each vertex in $S$ is directly connected to all the vertices
 in $H_1 \cup H_2$. Thus, the graph $H$ has the structure $H \cong K_{k} \vee \left(K_{k_1} \cup K_{n-k-k_1} \right)$, where $\vee$ denotes the
 join of graphs. This structure ensures maximal connectivity and pairwise proximity, thereby maximizing the hyper-Zagreb
 index. If $k_1=1$, then we have $H\cong \left(K_1 \cup K_{n-k-1}\right) \vee K_{k}$ and hence
 \begin{align*}
 HM(G)&\leq HM(H)=HM(K_{k} \vee \left(K_1 \cup K_{n-k-1}\right)\\
 &=4\,{k\choose 2}\,(n-1)^2+4\,{n-k-1\choose 2}\,(n-2)^2+k\,(n-k-1)\,(2n-3)^2+k\,(n+k-1)^2
 \end{align*}
 with equality iff $G\cong (K_1\cup K_{n-k-1})\vee K_k$. The result holds in (\ref{p1kd1}).

 \vspace*{3mm}

 Otherwise, $2\leq k_1\leq \lfloor\frac{n-k}{2}\rfloor$. One can easily see that the degree $d_i$ of vertex $v_i$ is given by:
 $$d_i =
 \begin{cases}
 n- 1 & \text{for } v_i \in S, \\[1mm]
 k_1 +k- 1 & \text{for } v_i \in V(H_1), \\[1mm]
 n-k_1-1 & \text{for } v_i \in V(H_2).
 \end{cases}$$
 Let $k_2=n-k-k_1$. Then $k_2\geq k_1$. We have $HM(G)\leq HM(H)=HM\Big((K_{k_1}\cup K_{k_2})\vee
 K_k\Big)$. We now prove the following claim.

 \vspace*{3mm}

 \noindent
 \begin{claim} \label{h1}
  $$HM\Big((K_{k_1-1}\cup K_{k_2+1})\vee K_k\Big)>HM\Big((K_{k_1}\cup K_{k_2})\vee K_k\Big).$$
 \end{claim}

 \noindent
 {\bf Proof of Claim \ref{h1}.} Now,
 \begin{align}
 &HM\Big((K_{k_1-1}\cup K_{k_2+1})\vee K_k\Big)-HM\Big((K_{k_1}\cup K_{k_2})\vee K_k\Big)\nonumber\\[2mm]
 =&4\,{k_1-1\choose 2}\,(k+k_1-2)^2+4\,{k_2+1\choose 2}\,(k+k_2)^2+k\,(k_1-1)\,(2n-k_2-3)^2+k\,(k_2+1)\nonumber\\[2mm]
 &~~~~~~~~~~~~~~\times (2n-k_1-1)^2-4\,{k_1\choose 2}\,(k+k_1-1)^2-4\,{k_2\choose 2}\,(k+k_2-1)^2\nonumber\\[2mm]
 &~~~~~~~~~~~~~~-k\,k_1\,(2n-k_2-2)^2-k\,k_2\,(2n-k_1-2)^2\nonumber\\[2mm]
 =&2\,(k_1-1)\,\Big[(k_1-2)\,(k+k_1-2)^2-k_1\,(k+k_1-1)^2\Big]+2\,k_2\,\Big[(k_2+1)\,(k+k_2)^2-(k_2-1)\nonumber\\[2mm]
 &(k+k_2-1)^2\Big]+k\,\Big[(k_1-1)\,(2n-k_2-3)^2-k_1\,(2n-k_2-2)^2\Big]+k\,\Big[(k_2+1)\,(2n-k_1-1)^2\nonumber\\[2mm]
 &~~~~~~~~~~~~~~~-k_2\,(2n-k_1-2)^2\Big]\nonumber\\[2mm]
 =&2\,(k_1-1)\,\Big[-2\,(k+k_1-2)^2-k_1\,(2k+2k_1-3)\Big]+2\,k_2\,\Big[2\,(k+k_2)^2+(k_2-1)\,(2k+2k_2-1)\Big]\nonumber\\[2mm]
 &+k\,\Big[-(2n-k_2-3)^2-k_1\,(4n-2k_2-5)+(2n-k_1-1)^2+k_2\,(4n-2k_1-3)\Big].\label{1sm1}
 \end{align}
 Since $k_2\geq k_1$, we obtain
 $$(2n-k_1-1)^2>(2n-k_2-3)^2,~\mbox{ and }~k_2\,(4n-2k_1-3)>k_1\,(4n-2k_2-5).$$
 Thus we have
 $$-(2n-k_2-3)^2-k_1\,(4n-2k_2-5)+(2n-k_1-1)^2+k_2\,(4n-2k_1-3)>0.$$
 Applying the above result in (\ref{1sm1}), we get
 \begin{align}
 &HM\Big((K_{k_1-1}\cup K_{k_2+1})\vee K_k\Big)-HM\Big((K_{k_1}\cup K_{k_2})\vee K_k\Big)\nonumber\\[2mm]
 >&2\,(k_1-1)\,\Big[-2\,(k+k_1-2)^2-k_1\,(2k+2k_1-3)\Big]+2\,k_2\,\Big[2\,(k+k_2)^2+(k_2-1)\,(2k+2k_2-1)\Big]\nonumber\\[2mm]
 >&2\,(k_1-1)\,\Big[-2\,(k+k_1-2)^2-k_1\,(2k+2k_1-3)+2\,(k+k_2)^2+(k_2-1)\,(2k+2k_2-1)\Big].\label{1sm2}
 \end{align}
 If $k_2\geq k_1+1$, then
 \begin{align*}
 &(k+k_2)^2>(k+k_1-2)^2,~~(k_2-1)\,(2k+2k_2-1)>k_1\,(2k+2k_1-3),~\mbox{ and hence }\\
 &-2\,(k+k_1-2)^2-k_1\,(2k+2k_1-3)+2\,(k+k_2)^2+(k_2-1)\,(2k+2k_2-1)>0.
 \end{align*}
 Using the above result, we obtain $HM\Big((K_{k_1-1}\cup K_{k_2+1})\vee K_k\Big)>HM\Big((K_{k_1}\cup K_{k_2})\vee
 K_k\Big)$. Otherwise, $k_2=k_1$. From (\ref{1sm2}), we obtain
 \begin{align*}
 &HM\Big((K_{k_1-1}\cup K_{k_2+1})\vee K_k\Big)-HM\Big((K_{k_1}\cup K_{k_2})\vee K_k\Big)\nonumber\\[2mm]
 >&2\,(k_1-1)\,\Big[-2\,(k+k_1-2)^2-k_1\,(2k+2k_1-3)+2\,(k+k_1)^2+(k_1-1)\,(2k+2k_1-1)\Big]\nonumber\\[2mm]
 =&2\,(k_1-1)\,(6k+8k_1-7)>0.
 \end{align*}
 Hence $HM\Big((K_{k_1-1}\cup K_{k_2+1})\vee K_k\Big)>HM\Big((K_{k_1}\cup K_{k_2})\vee
 K_k\Big)$. Thus, the {\bf Claim \ref{h1}} is proved.

 \vspace*{3mm}

 By employing the {\bf Claim \ref{h1}} several times, we get
    $$HM\Big((K_{k_1}\cup K_{k_2})\vee K_k\Big)<HM\Big((K_{k_1-1}\cup K_{k_2+1})\vee
    K_k\Big)<\cdots<HM\Big((K_1\cup K_{n-k-1})\vee K_k\Big).$$
 Utilizing the above result, we get
   $$HM(G)\leq HM\Big((K_{k_1}\cup K_{k_2})\vee K_k\Big)<HM\Big((K_1\cup K_{n-k-1})\vee K_k\Big).$$
 Thus, the theorem is proved.
 \end{proof}

 \begin{theorem} Let $G$ be a graph of order $n$ with edge connectivity $k'$. Then
 \begin{align}
 HM(G)\leq 4\,{k'\choose 2}\,(n-1)^2+4\,{n-k'-1\choose 2}\,(n-2)^2+k'\,(n-k'-1)\,(2n-3)^2+k'\,(n+k'-1)^2\label{1kd1}
 \end{align}
 with equality iff $G\cong (K_1\cup K_{n-k'-1})\vee K_{k'}$.
 \end{theorem}

 \begin{proof} We consider a function
 \begin{align*}
 h(y)&=4\,{y\choose 2}\,(n-1)^2+4\,{n-y-1\choose 2}\,(n-2)^2+y\,(n-y-1)\,(2n-3)^2+y\,(n+y-1)^2\\
 &=2\,y(y-1)\,(n-1)^2+2(n-y-1)\,(n-y-2)\,(n-2)^2+y\,(n-y-1)\,(2n-3)^2\\
 &~~~~~~~~~~~~~~~~~~~~~~~~~~~~~~~~~~~~~~~~+y\,(n+y-1)^2,~~y\leq k'.
 \end{align*}
 Then
 \begin{align*}
 h'(y)&=2\,(y-1)\,(n-1)^2+2y\,(n-1)^2-2\,(n-y-2)\,(n-2)^2-2\,(n-y-1)\,(n-2)^2\\
 &~~~~~~~~~~~~~~~~~+(n-y-1)\,(2n-3)^2-y\,(2n-3)^2+(n-y-1)^2+2y\,(n+y-1)\\
 &=2\,(2y-1)\,\Big[(n-1)^2-(n-1.5)^2\Big]+2\,(n-y-1)\,\Big[(n-1.5)^2-(n-2)^2\Big]\\
 &~~~~~+2\,(n-y-2)\,\Big[(n-1.5)^2-(n-2)^2\Big]+(n+y-1)^2+2y\,(n+y-1)\\
 &>0.
 \end{align*}
 So $h(y)$ is an increasing function on $y\leq k'$, and therefore
 $h(k)\leq h(k')$ as $k\leq k'$ (where $k$ is vertex connectivity). By Theorem \ref{pp1}, we get
 the result (\ref{1kd1}). Furthermore, the equality occurs iff $G\cong (K_1\cup K_{n-k'-1})\vee K_{k'}$.
 \end{proof}

\vspace*{6mm}

\section{Chemical applications of the $HM$ index.}
Topological indices are based on mathematical chemistry, so it is crucial to evaluate their chemical relevance in addition to their theoretical aspects. Researchers discovered that $HM$ strongly explains different physicochemical properties, including those of octanes, monocarboxylic acids, and the boiling point of benzenoid hydrocarbons.\cite{hayat2023,rajasekharaiah20,mirajkar21}
Recent studies have explored the chemical applications of topological indices.\cite{babai23,hayat24,tratnik18}

Building upon these findings, this section examines the chemical applicability of $HM$ by constructing linear, quadratic, and cubic regression models for standard enthalpy of vaporization ($DHVAP$) of nonane isomers (all C$_9$H$_{20}$). Table \ref{T11} presents the theoretical $HM$ values along with the $DHVAP$ for nonane isomers. The experimental $DHVAP$ values were collected from the literature.\cite{mondal2024}  Since all nonane isomers have identical molecular size ($n = 9$, $m = 8$), the observed structure–property relationships reflected genuine molecular structural information rather than size effects. To analyze the predictive ability of the $HM$ index, we use the following regression relations:
\begin{eqnarray}\label{EQN1}
Y=\ell_{1} X +\ell_{2},
\end{eqnarray}
\begin{eqnarray}\label{EQN2}
Y= m_{1} X^2 + m_{2}X + m_{3},
\end{eqnarray}
\begin{eqnarray}\label{EQN3}
Y= n_{1} X^3 + n_{2}X^2 + n_{3}X +n_{4},
\end{eqnarray}

where $Y$, $X$ represent property and descriptor, respectively, and $\ell_{1}$, $\ell_{2}$, $m_{1}$, $m_{2}$, $m_{3}$, $n_{1}$, $n_{2}$, $n_{3}$, and $n_{4}$ are fitting parameters. Along with relations \eqref{EQN1}, \eqref{EQN2}, and \eqref{EQN3}, we also take into account some other statistical parameters such as the coefficient of determination ($R^2$),  $F$-test ($F$), root mean square error ($RMSE$), and significance F ($SF$) to better understand the  regression relationships.

\begin{table}[H]
 \caption{Theoretical $HM$ and experimental $DHVAP$ for nonane isomers.}\label{T11}
\centering
\footnotesize
\begin{tabular}{lcccccccccccc}
\toprule
Nonanes&    $HM$&  $DHVAP$&    Nonanes&    $HM$&  $DHAVP$\\[0.7ex]
\midrule
C9:1&   114& 46.44&  C9:19&  192& 42.23\\[0.5ex]
C9:2&   130& 44.65&  C9:20&  170& 42.93\\[0.5ex]
C9:3&   132& 44.75&  C9:21&  166& 41.42\\[0.5ex]
C9:4&   132& 44.75&  C9:22&  188& 40.84\\[0.5ex]
C9:5&   168& 42.28&  C9:23&  194& 42.28\\[0.5ex]
C9:6&   150& 43.79&  C9:24&  152& 43.84\\[0.5ex]
C9:7&   148& 42.87&  C9:25&  150& 42.98\\[0.5ex]
C9:8&   148& 42.87&  C9:26&  176&  43.04\\[0.5ex]
C9:9&   146& 42.82&  C9:27&  154& 43.95\\[0.5ex]
C9:10&  172& 42.66&  C9:28&  234& 41\\[0.5ex]
C9:11&  152& 43.84&  C9:29&  208& 41\\[0.5ex]
C9:12&  150& 42.98&  C9:30&  222& 38.1\\[0.5ex]
C9:13&  172& 42.66&  C9:31&  212& 41.75\\[0.5ex]
C9:14&  134& 44.81&  C9:32&  192& 42.02\\[0.5ex]
C9:15&  134& 44.81&  C9:33&  196&  42.55\\[0.5ex]
C9:16&  190&  41.91&  C9:34&  170& 42.93\\[0.5ex]
C9:17&  186& 40.57&  C9:35&  180& 43.36\\[0.5ex]
C9:18&  184& 40.17&&&&&\\        [0.7ex]
 \bottomrule
\end{tabular}
\end{table}

\vspace*{5mm}

For $DHVAP$, the equation \eqref{EQN1} takes the following form:
\begin{eqnarray}\label{EQN4}
DHVAP = -0.0473(HM)+50.7080,~~~ R^{2}=0.7045,~~~RMSE=0.8563,\\
\nonumber F=78.6630,~~~SF=2.99 \times 10^{-10}.
\end{eqnarray}
Similarly, the equation \eqref{EQN2} takes the following form:
\begin{eqnarray}\label{EQN5}
DHVAP =0.000228(HM^2)-0.1260(HM)+57.3117,\\
\nonumber R^{2}=0.7227,~~~RMSE=0.8296,~~~F=41.6912,~~~SF=1.23\times 10^{-09}.
\end{eqnarray}
Finally, the equation \eqref{EQN3} takes the following form:
\begin{eqnarray}\label{EQN6}
DHVAP =-4.42\times 10^{-06}(HM^3)+0.002544(HM^2)-0.5218(HM)+79.3692,\\
\nonumber R^{2}=0.731,~~~RMSE=0.8170,~~~F=28.0775,~~~SF=5.71\times 10^{-09}.
\end{eqnarray}

\noindent
\vspace*{1mm}
Figure \ref{dh1} shows the linear fit for model \eqref{EQN4}, Figure \ref{dh2} displays the quadratic fit for model \eqref{EQN5}, and Figure \ref{dh3} shows the cubic fit for model \eqref{EQN6}. The blue, purple, and orange lines indicate the best-fit lines corresponding to the linear, quadratic, and cubic models, respectively, and the green circles indicate the ordered pairs $(x, y)$, where $x$ and $y$ are the $HM$ index and $DHVAP$ of nonane isomers, respectively.

To evaluate the performance of the $HM$ index in different regression models, we analyze relations \eqref{EQN4}--\eqref{EQN6} and Figures \ref{dh1}--\ref{dh3}. Based on the linear model \eqref{EQN4}, the $HM$ index accounts for $70\%$ of the variance in $DHVAP$, as visualized in Figure \ref{dh1}. The quadratic model \eqref{EQN5} shows that $HM$ explains $72\%$ of the variance in $DHVAP$, as illustrated in Figure \ref{dh2}. Moreover, the cubic model \eqref{EQN6} indicates that $HM$ explains $73\%$ of the variance in $DHVAP$, as visualized in Figure \ref{dh3}. The corresponding $SF$ values are much smaller than $0.05$, confirming the statistical significance of the models. Collectively, these results demonstrate that $HM$ effectively models $DHVAP$ for nonane isomers, with the observed structure–property relationships reflecting genuine structural information rather than size effects due to their identical molecular size ($n = 9$, $m = 8$).

\vspace*{1mm}
\begin{figure}[H]
\includegraphics[width=9cm]{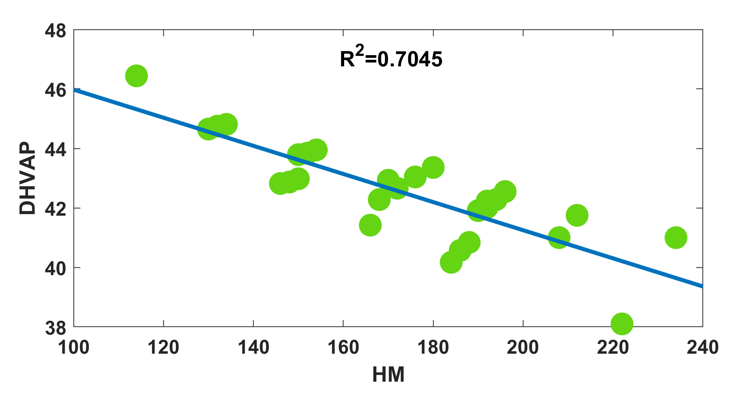}
\centering \caption{Linear fitting of $HM$ with $DHVAP$ for nonane isomers.} \label{dh1}
\end{figure}
 \vspace*{3mm}
\begin{figure}[H]
\includegraphics[width=9cm]{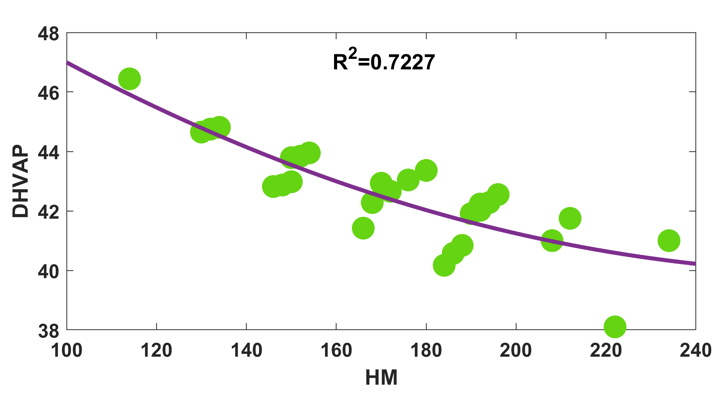}
\centering \caption{Quadratic fitting of $HM$ with $DHVAP$ for nonane isomers.} \label{dh2}
\end{figure}
\begin{figure}[H]
\includegraphics[width=9cm]{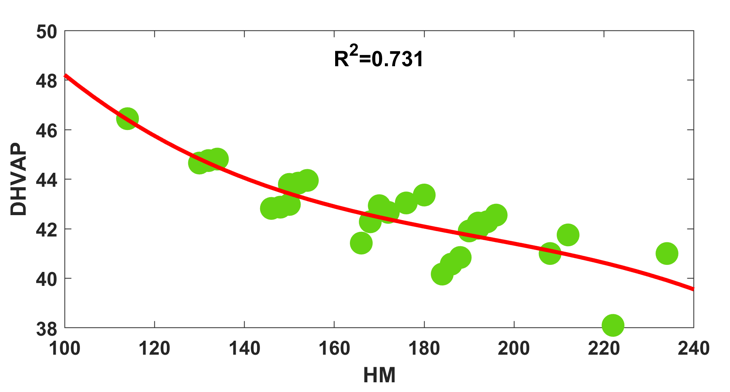}
\centering \caption{Cubic fitting of $HM$ with $DHVAP$ for
 nonane isomers.} \label{dh3}
\end{figure}
\vspace*{3mm}
\begin{table}[H]
\caption{Definitions and references of degree-based topological indices used in this study.}\label{T1000}
\centering
\begin{tabular}{lll}
\toprule
\textbf{Index} & \textbf{Definition} & \textbf{Reference} \\
\midrule
$M_1$ (first Zagreb) & $\displaystyle\sum_{v_iv_j\in E(G)} (d_i+d_j)$ & \cite{gut72} \\
$M_2$ (second Zagreb) & $\displaystyle\sum_{v_iv_j\in E(G)} d_i d_j$ & \cite{gut72} \\
$ALB$ (Albertson) & $\displaystyle\sum_{v_iv_j\in E(G)} |d_i-d_j|$ & \cite{albertson97} \\
$RM_2$ (modified second Zagreb) & $\displaystyle\sum_{v_iv_j\in E(G)} \frac{1}{d_i d_j}$ & \cite{milicevic04} \\
$PL$ (Platt) & $\displaystyle\sum_{v_iv_j\in E(G)} (d_i+d_j-2)$ & \cite{platt52,ali18} \\
$F$ (forgotten) & $\displaystyle\sum_{v_iv_j\in E(G)} (d_i^2+d_j^2)$ & \cite{furtula15} \\
$ISI$ (inverse sum indeg) & $\displaystyle\sum_{v_iv_j\in E(G)} \frac{d_i d_j}{d_i+d_j}$ & \cite{vukicevic10,vukicevic10b} \\
$HM$ (hyper-Zagreb) & $\displaystyle\sum_{v_iv_j\in E(G)} (d_i+d_j)^2$ & \cite{shirdel13} \\
\bottomrule
\end{tabular}
\end{table}
\vspace*{3mm}
\noindent
\begin{figure}[htp]
\centering
\begin{tikzpicture}[x=0.75cm,y=1.1cm]
\clip(-0.3,-1.56) rectangle (22.46,16.3);
\draw [line width=1.pt] (0.9,14.98)-- (7.9,14.98);
\draw [line width=1.pt] (9.,15.)-- (15.,15.);
\draw [line width=1.pt] (10.,16.)-- (10.,15.);
\draw [line width=1.pt] (16.,15.)-- (22.,15.);
\draw [line width=1.pt] (18.,16.)-- (18.,15.);
\draw (3.8,14.46) node[anchor=north west] {C8:1};
\draw (11.4,14.48) node[anchor=north west] {C8:2};
\draw (18.4,14.44) node[anchor=north west] {C8:3};
\draw [line width=1.pt] (10.,12.)-- (15.,12.);
\draw [line width=1.pt] (16.,12.)-- (21.,12.);
\draw [line width=1.pt] (12.,13.)-- (12.,12.);
\draw [line width=1.pt] (13.,13.)-- (12.,13.);
\draw [line width=1.pt] (17.,13.)-- (17.,12.);
\draw [line width=1.pt] (17.,11.)-- (17.,12.);
\draw (4.3,11.26) node[anchor=north west] {C8:4};
\draw (11.9,11.26) node[anchor=north west] {C8:5};
\draw (17.8,11.28) node[anchor=north west] {C8:6};
\draw [line width=1.pt] (2.,12.)-- (8.,12.);
\draw [line width=1.pt] (5.,13.)-- (5.,12.);
\draw [line width=1.pt] (3.,9.)-- (8.,9.);
\draw [line width=1.pt] (4.,10.)-- (4.,9.);
\draw [line width=1.pt] (9.,9.)-- (14.,9.);
\draw [line width=1.pt] (10.,9.)-- (10.,10.);
\draw [line width=1.pt] (12.,9.)-- (12.,8.);
\draw [line width=1.pt] (5.,9.)-- (5.,8.);
\draw [line width=1.pt] (15.,9.)-- (20.,9.);
\draw [line width=1.pt] (16.,10.)-- (16.,9.);
\draw [line width=1.pt] (19.,9.)-- (19.,8.);
\draw (5.5,8.26) node[anchor=north west] {C8:7};
\draw (10,8.26) node[anchor=north west] {C8:8};
\draw (16.84,8.26) node[anchor=north west] {C8:9};
\draw [line width=1.pt] (3.,6.)-- (8.,6.);
\draw [line width=1.pt] (9.,6.)-- (14.,6.);
\draw [line width=1.pt] (15.,6.)-- (19.,6.);
\draw [line width=1.pt] (5.,7.)-- (5.,6.);
\draw [line width=1.pt] (5.,6.)-- (5.,5.);
\draw [line width=1.pt] (11.,6.)-- (11.,7.);
\draw [line width=1.pt] (12.,6.)-- (12.,5.);
\draw [line width=1.pt] (16.,7.)-- (16.,6.);
\draw [line width=1.pt] (17.,6.)-- (17.,5.);
\draw [line width=1.pt] (16.,5.)-- (17.,5.);
\draw (5.36,5.28) node[anchor=north west] {C8:10};
\draw (9.9,5.28) node[anchor=north west] {C8:11};
\draw (17.32,5.28) node[anchor=north west] {C8:12};
\draw [line width=1.pt] (4.,3.)-- (8.,3.);
\draw [line width=1.pt] (9.,3.)-- (13.,3.);
\draw [line width=1.pt] (14.,3.)-- (18.,3.);
\draw [line width=1.pt] (6.,4.)-- (6.,3.);
\draw [line width=1.pt] (6.,3.)-- (6.,2.);
\draw [line width=1.pt] (7.,2.)-- (6.,2.);
\draw [line width=1.pt] (10.,3.)-- (10.,4.);
\draw [line width=1.pt] (10.,3.)-- (10.,2.);
\draw [line width=1.pt] (11.,3.)-- (11.,2.);
\draw [line width=1.pt] (15.,3.)-- (15.,4.);
\draw [line width=1.pt] (15.,3.)-- (15.,2.);
\draw [line width=1.pt] (17.,3.)-- (17.,2.);
\draw (4,2.26) node[anchor=north west] {C8:13};
\draw (11.3,2.26) node[anchor=north west] {C8:14};
\draw (15.2,2.26) node[anchor=north west] {C8:15};
\draw [line width=1.pt] (4.,0.)-- (8.,0.);
\draw [line width=1.pt] (9.,0.)-- (13.,0.);
\draw [line width=1.pt] (14.,0.)-- (17.,0.);
\draw [line width=1.pt] (5.,0.)-- (5.,1.);
\draw [line width=1.pt] (6.,-1.)-- (6.,0.);
\draw [line width=1.pt] (6.,1.)-- (6.,0.);
\draw [line width=1.pt] (10.,0.)-- (10.,1.);
\draw [line width=1.pt] (11.,0.)-- (11.,-1.);
\draw [line width=1.pt] (12.,0.)-- (12.,1.);
\draw [line width=1.pt] (15.,0.)-- (15.,1.);
\draw [line width=1.pt] (15.,0.)-- (15.,-1.);
\draw [line width=1.pt] (16.,0.)-- (16.,1.);
\draw [line width=1.pt] (16.,0.)-- (16.,-1.);
\draw (6.32,-0.74) node[anchor=north west] {C8:16};
\draw (11.3,-0.76) node[anchor=north west] {C8:17};
\draw (16.34,-0.72) node[anchor=north west] {C8:18};
\begin{scriptsize}
\draw [ball color=AmberDawn] (0.9,14.98) circle (2.5pt);
\draw [ball color=AmberDawn] (7.9,14.98) circle (2.5pt);
\draw [ball color=AmberDawn] (1.9,14.98) circle (2.5pt);
\draw [ball color=AmberDawn] (2.9,14.98) circle (2.5pt);
\draw [ball color=AmberDawn] (3.9,14.98) circle (2.5pt);
\draw [ball color=AmberDawn] (4.9,14.98) circle (2.5pt);
\draw [ball color=AmberDawn] (5.9,14.98) circle (2.5pt);
\draw [ball color=AmberDawn] (6.9,14.98) circle (2.5pt);
\draw [ball color=AmberDawn] (9.,15.) circle (2.5pt);
\draw [ball color=AmberDawn] (15.,15.) circle (2.5pt);
\draw [ball color=AmberDawn] (10.,15.) circle (2.5pt);
\draw [ball color=AmberDawn] (11.,15.) circle (2.5pt);
\draw [ball color=AmberDawn] (12.,15.) circle (2.5pt);
\draw [ball color=AmberDawn] (13.,15.) circle (2.5pt);
\draw [ball color=AmberDawn] (14.,15.) circle (2.5pt);
\draw [ball color=AmberDawn] (10.,16.) circle (2.5pt);
\draw [ball color=AmberDawn] (16.,15.) circle (2.5pt);
\draw [ball color=AmberDawn] (22.,15.) circle (2.5pt);
\draw [ball color=AmberDawn] (17.,15.) circle (2.5pt);
\draw [ball color=AmberDawn] (18.,15.) circle (2.5pt);
\draw [ball color=AmberDawn] (19.,15.) circle (2.5pt);
\draw [ball color=AmberDawn] (20.,15.) circle (2.5pt);
\draw [ball color=AmberDawn] (21.,15.) circle (2.5pt);
\draw [ball color=AmberDawn] (18.,16.) circle (2.5pt);
\draw [ball color=AmberDawn] (10.,12.) circle (2.5pt);
\draw [ball color=AmberDawn] (15.,12.) circle (2.5pt);
\draw [ball color=AmberDawn] (11.,12.) circle (2.5pt);
\draw [ball color=AmberDawn] (12.,12.) circle (2.5pt);
\draw [ball color=AmberDawn] (13.,12.) circle (2.5pt);
\draw [ball color=AmberDawn] (14.,12.) circle (2.5pt);
\draw [ball color=AmberDawn] (16.,12.) circle (2.5pt);
\draw [ball color=AmberDawn] (21.,12.) circle (2.5pt);
\draw [ball color=AmberDawn] (17.,12.) circle (2.5pt);
\draw [ball color=AmberDawn] (18.,12.) circle (2.5pt);
\draw [ball color=AmberDawn] (19.,12.) circle (2.5pt);
\draw [ball color=AmberDawn] (20.,12.) circle (2.5pt);
\draw [ball color=AmberDawn] (12.,13.) circle (2.5pt);
\draw [ball color=AmberDawn] (13.,13.) circle (2.5pt);
\draw [ball color=AmberDawn] (17.,13.) circle (2.5pt);
\draw [ball color=AmberDawn] (17.,11.) circle (2.5pt);
\draw [ball color=AmberDawn] (2.,12.) circle (2.5pt);
\draw [ball color=AmberDawn] (8.,12.) circle (2.5pt);
\draw [ball color=AmberDawn] (5.,13.) circle (2.5pt);
\draw [ball color=AmberDawn] (5.,12.) circle (2.5pt);
\draw [ball color=AmberDawn] (3.,12.) circle (2.5pt);
\draw [ball color=AmberDawn] (4.,12.) circle (2.5pt);
\draw [ball color=AmberDawn] (6.,12.) circle (2.5pt);
\draw [ball color=AmberDawn] (7.,12.) circle (2.5pt);
\draw [ball color=AmberDawn] (3.,9.) circle (2.5pt);
\draw [ball color=AmberDawn] (8.,9.) circle (2.5pt);
\draw [ball color=AmberDawn] (4.,10.) circle (2.5pt);
\draw [ball color=AmberDawn] (4.,9.) circle (2.5pt);
\draw [ball color=AmberDawn] (5.,9.) circle (2.5pt);
\draw [ball color=AmberDawn] (6.,9.) circle (2.5pt);
\draw [ball color=AmberDawn] (7.,9.) circle (2.5pt);
\draw [ball color=AmberDawn] (9.,9.) circle (2.5pt);
\draw [ball color=AmberDawn] (14.,9.) circle (2.5pt);
\draw [ball color=AmberDawn] (10.,9.) circle (2.5pt);
\draw [ball color=AmberDawn] (11.,9.) circle (2.5pt);
\draw [ball color=AmberDawn] (12.,9.) circle (2.5pt);
\draw [ball color=AmberDawn] (13.,9.) circle (2.5pt);
\draw [ball color=AmberDawn] (10.,10.) circle (2.5pt);
\draw [ball color=AmberDawn] (12.,8.) circle (2.5pt);
\draw [ball color=AmberDawn] (5.,8.) circle (2.5pt);
\draw [ball color=AmberDawn] (15.,9.) circle (2.5pt);
\draw [ball color=AmberDawn] (20.,9.) circle (2.5pt);
\draw [ball color=AmberDawn] (16.,10.) circle (2.5pt);
\draw [ball color=AmberDawn] (16.,9.) circle (2.5pt);
\draw [ball color=AmberDawn] (19.,9.) circle (2.5pt);
\draw [ball color=AmberDawn] (19.,8.) circle (2.5pt);
\draw [ball color=AmberDawn] (17.,9.) circle (2.5pt);
\draw [ball color=AmberDawn] (18.,9.) circle (2.5pt);
\draw [ball color=AmberDawn] (3.,6.) circle (2.5pt);
\draw [ball color=AmberDawn] (8.,6.) circle (2.5pt);
\draw [ball color=AmberDawn] (9.,6.) circle (2.5pt);
\draw [ball color=AmberDawn] (14.,6.) circle (2.5pt);
\draw [ball color=AmberDawn] (15.,6.) circle (2.5pt);
\draw [ball color=AmberDawn] (19.,6.) circle (2.5pt);
\draw [ball color=AmberDawn] (4.,6.) circle (2.5pt);
\draw [ball color=AmberDawn] (5.,6.) circle (2.5pt);
\draw [ball color=AmberDawn] (6.,6.) circle (2.5pt);
\draw [ball color=AmberDawn] (7.,6.) circle (2.5pt);
\draw [ball color=AmberDawn] (10.,6.) circle (2.5pt);
\draw [ball color=AmberDawn] (11.,6.) circle (2.5pt);
\draw [ball color=AmberDawn] (12.,6.) circle (2.5pt);
\draw [ball color=AmberDawn] (13.,6.) circle (2.5pt);
\draw [ball color=AmberDawn] (16.,6.) circle (2.5pt);
\draw [ball color=AmberDawn] (17.,6.) circle (2.5pt);
\draw [ball color=AmberDawn] (18.,6.) circle (2.5pt);
\draw [ball color=AmberDawn] (5.,7.) circle (2.5pt);
\draw [ball color=AmberDawn] (5.,5.) circle (2.5pt);
\draw [ball color=AmberDawn] (11.,7.) circle (2.5pt);
\draw [ball color=AmberDawn] (12.,5.) circle (2.5pt);
\draw [ball color=AmberDawn] (16.,7.) circle (2.5pt);
\draw [ball color=AmberDawn] (17.,5.) circle (2.5pt);
\draw [ball color=AmberDawn] (16.,5.) circle (2.5pt);
\draw [ball color=AmberDawn] (4.,3.) circle (2.5pt);
\draw [ball color=AmberDawn] (8.,3.) circle (2.5pt);
\draw [ball color=AmberDawn] (9.,3.) circle (2.5pt);
\draw [ball color=AmberDawn] (13.,3.) circle (2.5pt);
\draw [ball color=AmberDawn] (14.,3.) circle (2.5pt);
\draw [ball color=AmberDawn] (18.,3.) circle (2.5pt);
\draw [ball color=AmberDawn] (5.,3.) circle (2.5pt);
\draw [ball color=AmberDawn] (6.,3.) circle (2.5pt);
\draw [ball color=AmberDawn] (7.,3.) circle (2.5pt);
\draw [ball color=AmberDawn] (10.,3.) circle (2.5pt);
\draw [ball color=AmberDawn] (11.,3.) circle (2.5pt);
\draw [ball color=AmberDawn] (12.,3.) circle (2.5pt);
\draw [ball color=AmberDawn] (15.,3.) circle (2.5pt);
\draw [ball color=AmberDawn] (16.,3.) circle (2.5pt);
\draw [ball color=AmberDawn] (17.,3.) circle (2.5pt);
\draw [ball color=AmberDawn] (6.,4.) circle (2.5pt);
\draw [ball color=AmberDawn] (6.,2.) circle (2.5pt);
\draw [ball color=AmberDawn] (7.,2.) circle (2.5pt);
\draw [ball color=AmberDawn] (10.,4.) circle (2.5pt);
\draw [ball color=AmberDawn] (10.,2.) circle (2.5pt);
\draw [ball color=AmberDawn] (11.,2.) circle (2.5pt);
\draw [ball color=AmberDawn] (15.,4.) circle (2.5pt);
\draw [ball color=AmberDawn] (15.,2.) circle (2.5pt);
\draw [ball color=AmberDawn] (17.,2.) circle (2.5pt);
\draw [ball color=AmberDawn] (4.,0.) circle (2.5pt);
\draw [ball color=AmberDawn] (8.,0.) circle (2.5pt);
\draw [ball color=AmberDawn] (5.,0.) circle (2.5pt);
\draw [ball color=AmberDawn] (6.,0.) circle (2.5pt);
\draw [ball color=AmberDawn] (7.,0.) circle (2.5pt);
\draw [ball color=AmberDawn] (9.,0.) circle (2.5pt);
\draw [ball color=AmberDawn] (13.,0.) circle (2.5pt);
\draw [ball color=AmberDawn] (14.,0.) circle (2.5pt);
\draw [ball color=AmberDawn] (17.,0.) circle (2.5pt);
\draw [ball color=AmberDawn] (10.,0.) circle (2.5pt);
\draw [ball color=AmberDawn] (11.,0.) circle (2.5pt);
\draw [ball color=AmberDawn] (12.,0.) circle (2.5pt);
\draw [ball color=AmberDawn] (15.,0.) circle (2.5pt);
\draw [ball color=AmberDawn] (16.,0.) circle (2.5pt);
\draw [ball color=AmberDawn] (5.,1.) circle (2.5pt);
\draw [ball color=AmberDawn] (6.,-1.) circle (2.5pt);
\draw [ball color=AmberDawn] (6.,1.) circle (2.5pt);
\draw [ball color=AmberDawn] (10.,1.) circle (2.5pt);
\draw [ball color=AmberDawn] (11.,-1.) circle (2.5pt);
\draw [ball color=AmberDawn] (12.,1.) circle (2.5pt);
\draw [ball color=AmberDawn] (15.,1.) circle (2.5pt);
\draw [ball color=AmberDawn] (15.,-1.) circle (2.5pt);
\draw [ball color=AmberDawn] (16.,1.) circle (2.5pt);
\draw [ball color=AmberDawn] (16.,-1.) circle (2.5pt);

\end{scriptsize}
\end{tikzpicture}
\caption{Molecular graphs of octane isomers.}
\label{OCTANE}
\end{figure}
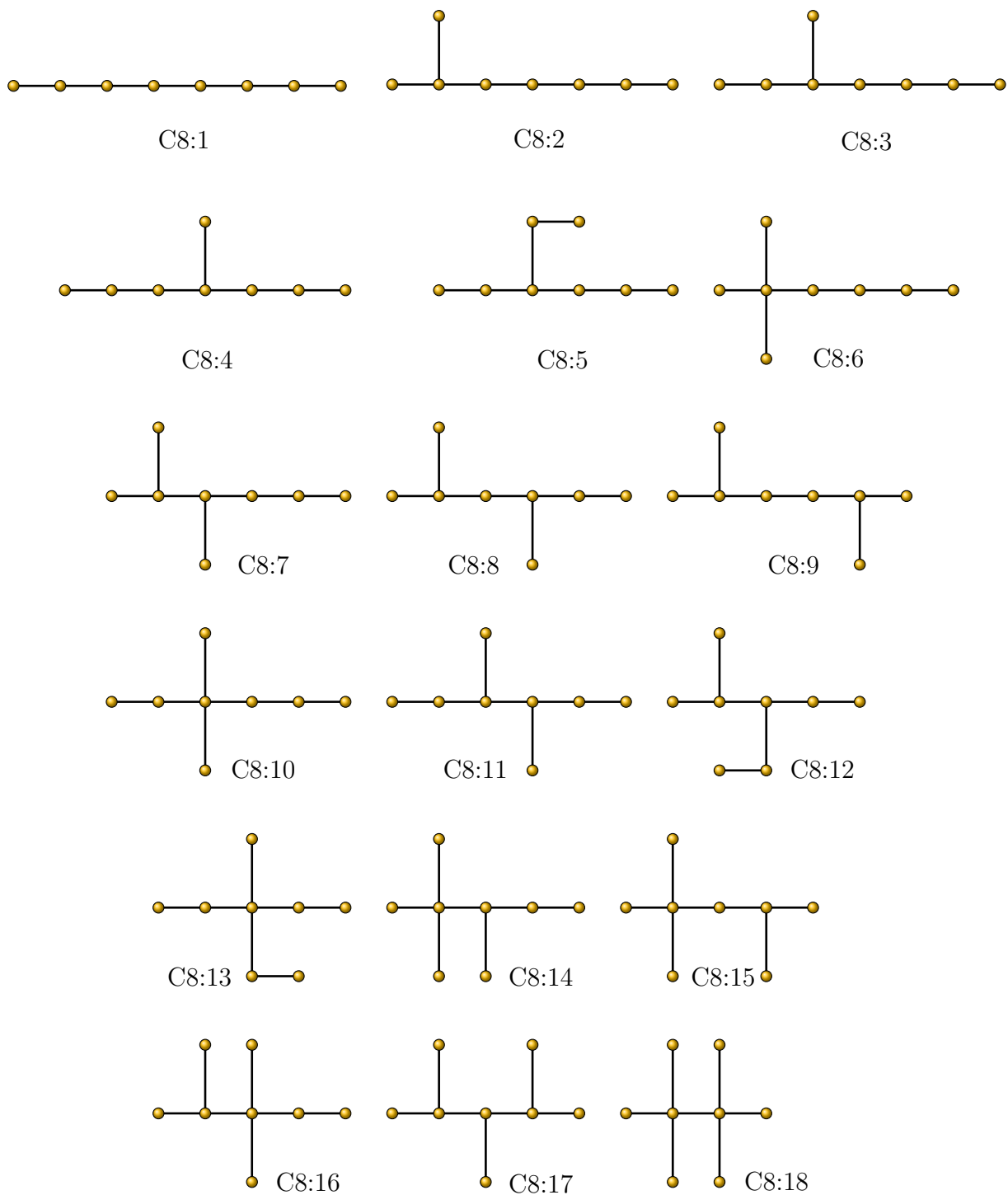


\section{Degeneracy}
Topological indices are essential for establishing structure-property  relationships as well as for distinguishing isomers. The ability of a topological index to discriminate between isomers is crucial for the encoding and computational processing of chemical structures.
Konstantinova introduced sensitivity ($ST$) to measure discriminative capability using the definition:
\begin{equation}
S_A = \frac{N - N_A}{N},
\end{equation}
where $N$ and $N_A$ represent the total number of isomers and the total number of isomers that the index $A$ cannot differentiate, respectively.\cite{kon96} An $S_A$ value closer to 1 indicates a higher discriminative power of the index.

\vspace*{2mm}
\begin{table}[H]
 \caption{The topological indices ($HM$, $RM_2$, $M_1$, $M_2$, $F$, $PL$, $ALB$, $ISI$) values for octanes.}\label{T100}
\centering
\small
\begin{tabular}{lcccccccc}
\toprule
Octanes& $HM$& $RM_2$& $M_1$& $M_2$& $F$&  $PL$& $ALB$& $ISI$\\
\midrule
C8:1& 98 & 2.25 & 26 & 24 & 50  & 12 & 2& 6.3333\\[0.9ex]
C8:2& 114 & 2.0833 & 28 & 26 & 62   & 14 & 6& 6.3667\\[0.9ex]
C8:3 &116& 2.1667 & 28 & 27 & 62   & 14 & 6& 6.4833\\[0.9ex]
C8:4&  116 & 2.1667 & 28 & 27 & 62   & 14 & 6& 6.4833\\[0.9ex]
C8:5&  118 & 2.25 & 28 & 28 & 62  & 14 & 6& 6.6\\[0.9ex]
C8:6& 152 & 1.875 & 32 & 30 & 92  & 18 & 12& 6.4\\[0.9ex]
C8:7& 134 & 2.0278 & 30 & 30 & 74 & 16 & 8& 6.6167\\[0.9ex]
C8:8& 132 & 2 & 30 & 29 & 74  & 16 & 10& 6.5167\\[0.9ex]
C8:9& 130 & 1.9167 & 30 & 28 & 74  & 16 & 10& 6.4\\[0.9ex]
C8:10& 156 & 2 & 32 & 32 & 92  & 18 & 12& 6.6\\[0.9ex]
C8:11& 136 & 2.1111 & 30 & 31 & 74   & 16 & 8& 6.7333\\[0.9ex]
C8:12& 136 & 2.1111 & 30 & 31 & 74   & 16 & 8& 6.7333\\[0.9ex]
C8:13& 160 & 2.125 & 32 & 34 & 92   & 18 & 12& 6.8\\[0.9ex]
C8:14& 174 & 1.8333 & 34 & 35 & 104   & 20 & 14& 6.7310\\[0.9ex]
C8:15& 168 & 1.7083 & 34 & 32 & 104  & 20 & 16& 6.4333\\[0.9ex]
C8:16& 176 & 1.875 & 34 & 36 & 104   & 20 & 14& 6.8143\\[0.9ex]
C8:17& 152 & 1.8889 & 32 & 33 & 86   & 18 & 10& 6.75\\[0.9ex]
C8:18& 214 & 1.5625 & 38 & 40 & 134  & 24  & 18& 6.8\\[0.9ex]
\bottomrule
\end{tabular}
\end{table}

\vspace*{2mm}
In this investigation, we evaluate the discriminative capabilities of the $HM$ index alongside several established degree-based topological indices using octane isomers. Table \ref{T1000} lists the definitions and references for all considered indices. 
Table \ref{T100} presents the calculated values of these indices for all octane isomers, 
and Figure \ref{OCTANE} illustrates their molecular graphs.

 Figure~\ref{DEG} illustrates the discriminative sensitivity values calculated using Konstantinova's formula. The $HM$ index exhibits the highest sensitivity among all indices considered, with a value of $0.833$. Indices such as $M_2$, $RM_2$, and $ISI$ show moderately lower values ($0.722$), while $M_1$, $PL$, $F$, and $ALB$ exhibit substantially lower discriminative power ($0.333$--$0.444$). These observations indicate that the $HM$ index possesses strong structural discrimination capability and performs favorably compared with the other indices considered for distinguishing between isomeric molecular structures in QSPR studies.
 \vspace*{2mm}
 \begin{figure}[H]
\includegraphics[width=13cm]{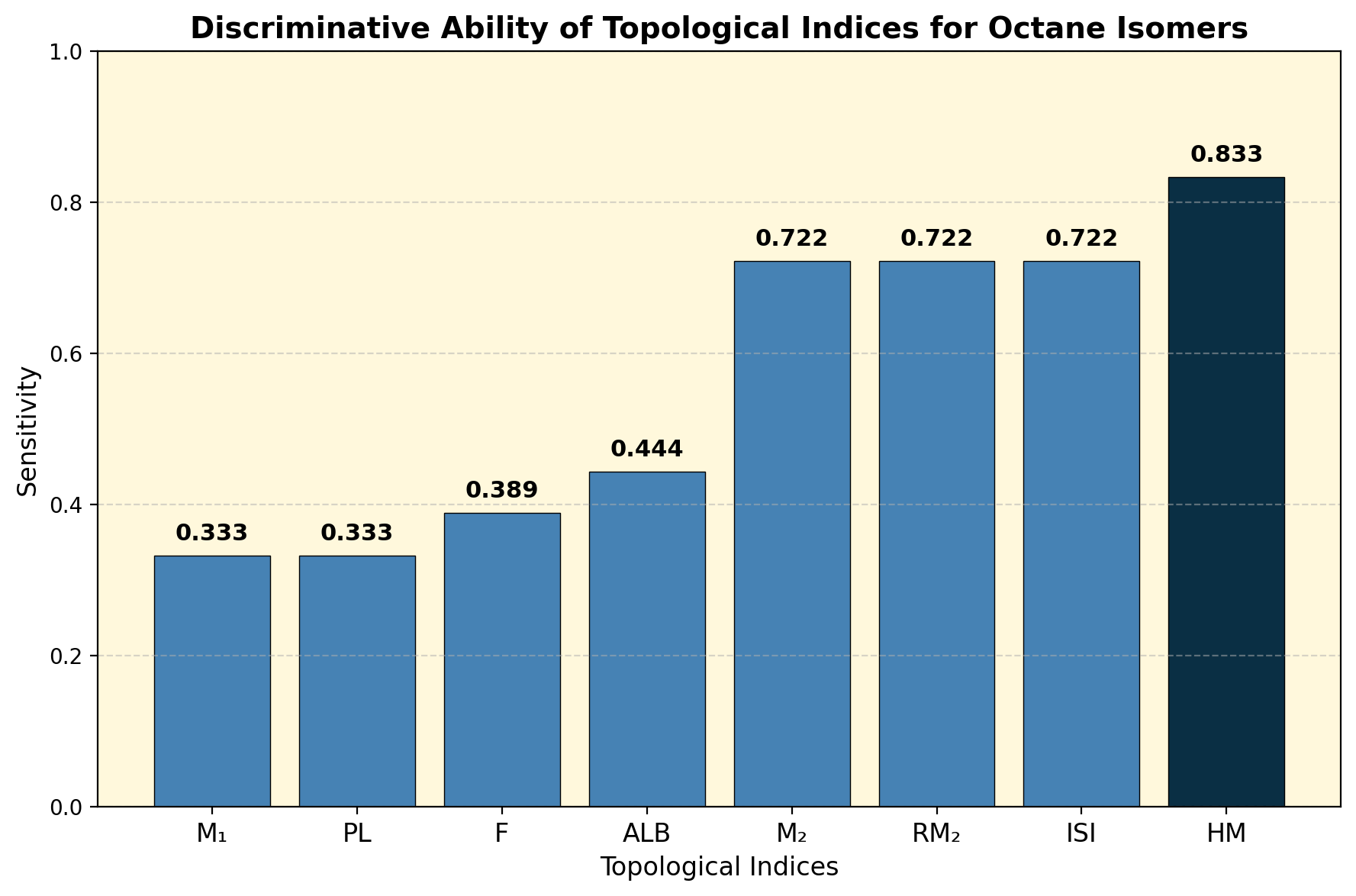 }
\centering \caption{Discriminative ability of topological indices for octanes.} \label{DEG}
\end{figure}

\section{Conclusions}
 In this paper, we analyzed the maximal graphs for $HM$ among all connected graphs, under constraints on the graph order along with a specified vertex-connectivity or edge-connectivity. These problems were previously addressed as open challenges in the literature.\cite{hayat2023} However, determining the minimum $HM$ index and characterizing the minimal graphs among connected graphs with fixed order and given vertex or edge connectivity remain open problems. We also investigated graphs with a specified number of leaves, chromatic number, or independence number that maximize $HM$. Furthermore, the applicability of $HM$ was investigated through the analysis of its structure–property relationships.We observed that $HM$ effectively predicts the $DHVAP$ of nonane isomers through linear, quadratic, and cubic regression models. Additionally, the discriminative power of $HM$ compares favorably with several established degree-based indices when applied to octane isomers.









\vspace*{4mm}

\baselineskip=0.2in

\begin{thebibliography}{99}

\bibitem{wiener47} H. Wiener, Structural determination of paraffin boiling points, {\it J. Am. Chem. Soc.} {\bf 69} (1947) 17--20.
\bibitem{das1666} K.C.~Das, Comparison between Zagreb eccentricity indices and the eccentric connectivity index, the second geometric--arithmetic index and the Graovac--Ghorbani index, {\it Croat.\ Chem.\ Acta} {\bf 89} (2016) 505--510.

\bibitem{dong11}
H.~Dong, B.~Zhou, N.~Trinajstić,
A novel version of the edge--Szeged index,
{\it Croat.\ Chem.\ Acta} {\bf 84} (2011) 543--545.

\bibitem{Heliyon2020}M.C. Shanmukha, N.S. Basavarajappa, K.C. Shilpa, A. Usha, Degree-based topological indices on anticancer drugs with QSPR analysis, {\it Heliyon} {\bf 6} (2020) e04235.

\bibitem{lucic09} B. Lu\v{c}i\'c, N. Trinajsti\'c, B. Zhou, Comparison between the sum-connectivity index and product-connectivity index for benzenoid hydrocarbons, {\it Chem. Phys. Lett.} {\bf 475} (2009) 146--148.

\bibitem{Shanmukha244} M.C. Shanmukha, A. Usha, V.R. Kulli, K.C. Shilpa, Chemical applicability and curvilinear regression models of vertex-degree-based topological index: Elliptic Sombor index, {\it Int. J. Quantum Chem.} {\bf 124} (2024) e27376.


\bibitem{gut72} I. Gutman, N. Trinajsti\'{c}, Graph theory and molecular orbitals. Total $\pi$-electron energy of alternant hydrocarbons, {\it Chem. Phys. Lett.\/} {\bf 17} (1972) 535--538.

\bibitem{liu066} B. Liu, I. Gutman, Upper bounds for Zagreb indices of connected graphs, {\it MATCH Commun. Math. Comput. Chem.\/} {\bf 55} (2006) 439--446.

\bibitem{das044} K.C. Das, I. Gutman, Some properties of the second Zagreb index, {\it MATCH Commun. Math. Comput. Chem.\/} {\bf 52} (2004) 103--112.

\bibitem{das100} K.C. Das, On comparing Zagreb indices of graphs, {\it MATCH Commun. Math. Comput. Chem.\/} {\bf 63}  (2010) 433--440.

\bibitem{gut15} I. Gutman, M. Kamran Jamil, N. Akhter, Graphs with fixed number of pendent vertices and minimal first Zagreb index, {\it Trans. Comb.\/} {\bf 4}  (2015) 43--48.

\bibitem{liu177} Z. Liu, Y. Chen, S. Li, The first Zagreb index, vertex-connectivity, minimum degree and independent number in graphs, {\it Int. J. Math. Combin.\/} {\bf 2} (2017) 34--42.
\bibitem{gut04} I. Gutman, K.C. Das, The first Zagreb index 30 years after, {\it MATCH Commun. Math. Comput. Chem.\/} {\bf 50} (2004) 83--92.


\bibitem{gutman13} I. Gutman, M. Goubko, Trees with fixed number of pendent vertices with minimal first Zagreb index, {\it Bull. Int. Math. Virtual Inst.\/} {\bf 3} (2013) 167--175.


\bibitem{liu102} M. Liu, B. Liu, The second Zagreb indices and Wiener polarity indices of trees with given degree sequence, {\it MATCH Commun. Math. Comput. Chem.\/} {\bf 67} (2012) 439--450.
\bibitem{nikolic03} S. Nikoli\'{c}, G. Kova\v{c}evi\'{c}, A. Mili\v{c}evi\'{c}, N. Trinajsti\'{c}, The Zagreb indices 30 years after, {\it Croat. Chem. Acta\/} {\bf 76}(2) (2003) 113--124.
\bibitem{xu11} K. Xu, The Zagreb indices of graphs with a given clique number, {\it Appl. Math. Lett.\/} {\bf 24} (2011) 1026--1030.
\bibitem{shirdel13} G.H. Shirdel, H. Rezapour, A.M. Sayadi, The hyper-Zagreb index of graph operations, {\it Iranian J. Math. Chem.\/} {\bf 4} (2013) 213--220.
\bibitem{mirajkar21} K.G. Mirajkar, B.R. Doddamani, H.B. Huchesh, On correlation of physicochemical properties and the hyper Zagreb index for some molecular structures, {\it South East Asian J. Math. Math. Sci.\/} {\bf 17} (3) (2021) 331--346.
\bibitem{hayat2023} S. Hayat, M.A. Khan, A. Khan, H. Jamil, M.Y.H. Malik, Extremal hyper-Zagreb index of trees of given segments with applications to regression modeling in QSPR studies, {\it Alexandria Eng. J.} {\bf 80} (2023) 259--268.

\bibitem{elumalai22} S. Elumalai, T. Mansour, A short note on Zagreb indices and hyper Zagreb indices of graphs, {\it Math. Rep.\/} {\bf 24} (2022).

\bibitem{elumalai18} S. Elumalai, T. Mansour, M.A. Rostami, New bounds on the hyper-Zagreb index for the simple connected graphs, {\it Electron. J. Graph Theory Appl.\/} {\bf 6} (2018) 166--177.
\bibitem{elumalai19} S. Elumalai, T. Mansour, M.A. Rostami, G.B.A. Xavier, A short note on hyper Zagreb index, {\it Bol. Soc. Paran. Mat.\/} {\bf 37} (2019) 51--58.


\bibitem{liu200} H. Liu, Z. Tang, Maximal hyper-Zagreb index of trees, unicyclic and bicyclic graphs with a given order and matching number, {\it Discrete Math. Lett.\/} {\bf 4} (2020) 11--18.


\bibitem{raza22} Z. Raza, S. Balachandran, S. Elumalai, A. Ali, On general sum–connectivity index of trees of fixed maximum degree and order, {\it MATCH Commun. Math. Comput. Chem.\/} {\bf 88} (2022) 643--658. 
\bibitem{vetrik} E. Swartz, T. Vetr\'{\i}k, General sum-connectivity index and general Randi\'{c} index of trees with given maximum degree, {\it Discrete Math. Lett.\/} {\bf 12} (2023) 181--188.
\bibitem{zhong20} L. Zhong, Q. Qian, The minimum general sum-connectivity index of trees with given matching number, {\it Bull. Malays. Math. Sci. Soc.\/} {\bf 43} (2020) 1527--1544.


\bibitem{vetrik23} T. Vetr\'{i}k, Degree-based function index for graphs with given diameter, {\it Discrete Appl. Math.\/} {\bf 333} (2023) 59--70.


\bibitem{walia24} M. Wali, J. Qian, C. Shi, Extremal graphs for vertex–degree–based indices with given matching number, {\it MATCH Commun. Math. Comput. Chem.\/} {\bf 91} (2024) 499--512.

\bibitem{rajasekharaiah20} G.V. Rajasekharaiah, U.P. Murthy, Hyper-Zagreb indices of graphs and its applications, {\it J. Algebra Comb. Discrete Appl.\/} {\bf 8}(1) (2020) 9--22.

\bibitem{babai23} A. Babai, S. Mondal, K.C. Das, Szeged indices of bicyclic graphs with applications as molecular descriptor, {\it Croat. Chem. Acta} {\bf 96} (2023) 153--170.

\bibitem{hayat24} S. Hayat, A. Khan, K. Ali, J.B. Liu, Structure-property modeling for thermodynamic properties of benzenoid hydrocarbons by temperature-based topological indices, {\it Ain Shams Eng. J.} {\bf 15} (2024) 102586.
\bibitem{tratnik18} N.~Tratnik, The Wiener polarity index of benzenoid systems and nanotubes,
{\it Croat.\ Chem.\ Acta} {\bf 91} (2018) 309--315.

\bibitem{mondal2024} S. Mondal, D. Huh, K.C. Das, Role of GA, AG and R in structure-property modelling, {\it Polycycl. Aromat. Comp.} (2024) e2405526.

\bibitem{kon96} E.V. Konstantinova, The discrimination ability of some topological and information distance indices for graphs of unbranched hexagonal systems, {\it J. Chem. Inf. Comput. Sci.\/} {\bf 36} (1996) 54--57.
\bibitem{albertson97} M.O. Albertson, The irregularity of a graph, {\it Ars Combin.\/} {\bf 46} (1997) 219--225.
\bibitem{milicevic04} A. Mili\'{c}evi\'{c}, S. Nikoli\'{c}, N. Trinajsti\'{c}, On reformulated Zagreb indices, {\it Mol. Divers.} {\bf 8} (2004) 393--399.
\bibitem{ali18} A. Ali, D. Dimitrov, On the extremal graphs with respect to bond incident degree indices, {\it Discrete Appl. Math.} {\bf 238} (2018) 32--40.
\bibitem{platt52}J.R. Platt, Prediction of isomeric differences in paraffin properties, {\it J. Phys. Chem.\/} {\bf 56} (1952) 328--336.
\bibitem{furtula15} B. Furtula, I. Gutman, A forgotten topological index, {\it J. Math. Chem.} {\bf 53} (2015) 1184--1190.
\bibitem{vukicevic10} D. Vukicevi\'c, M. Ga\v{s}perov, Bond additive modeling. 1. Adriatic indices, {\it Croat. Chem. Acta\/} {\bf 83} (2010) 243--260.
\bibitem{vukicevic10b} D. Vuki\v{c}evi\'c, Bond additive modeling 2. Mathematical properties of max-min rodeg index, {\it Croat. Chem. Acta\/} {\bf 83} (2010) 261--273.


\end{thebibliography}
\end{document}